\documentclass[runningheads,a4paper]{llncs}
\usepackage[square,sort,comma,numbers]{natbib}
\usepackage{amssymb,amsmath}
\usepackage[mathscr]{euscript}
\usepackage{graphicx}
\usepackage{pstricks}
\usepackage{pstricks-add}
\usepackage{fp}
\usepackage{natbib}
\usepackage{dsfont}
\usepackage{algorithm}
\usepackage{algpseudocode}
\usepackage{appendix}
\algrenewcommand\algorithmicfor{\textbf{Foreach}}
\algrenewcommand\algorithmicwhile{\textbf{While}}
\algrenewcommand\algorithmicdo{\textbf{Do}}
\usepackage{url}
\urldef{\mailsa}\path|{alfred.hofmann, ursula.barth, ingrid.haas, frank.holzwarth,|
\urldef{\mailsb}\path|anna.kramer, leonie.kunz, christine.reiss, nicole.sator,|
\urldef{\mailsc}\path|erika.siebert-cole, peter.strasser, lncs}@springer.com|    
\newcommand{\keywords}[1]{\par\addvspace\baselineskip
\noindent\keywordname\enspace\ignorespaces#1}

{\makeatletter
 \gdef\xxxmark{%
   \expandafter\ifx\csname @mpargs\endcsname\relax 
     \expandafter\ifx\csname @captype\endcsname\relax 
       \marginpar{xxx}
     \else
       xxx 
     \fi
   \else
     xxx 
   \fi}
 \gdef\xxx{\@ifnextchar[\xxx@lab\xxx@nolab}
 \long\gdef\xxx@lab[#1]#2{{\bf [\xxxmark #2 ---{\sc #1}]}}
 \long\gdef\xxx@nolab#1{{\bf [\xxxmark #1]}}
}

\newcommand{\yinew}{{y_{i,\text{new}}}}
\newcommand{\yiold}{{y_{i,\text{old}}}}

\begin{document}
\title{A Unified Approach to Online Allocation Algorithms via Randomized Dual Fitting}
\titlerunning{A Unified Approach to Online Allocation Algorithms via Randomized Dual Fitting}
\author{Rad Niazadeh$^{\dag}$ \and Robert D. Kleinberg$^{\dag}$
\thanks{Partially supported by NSF grant AF-0910940, AFOSR grant FA9550-09-1-0100, a Google Research Grant, and a Microsoft New Faculty Fellowship.}}
\authorrunning{A Unified Approach to the Online Allocation Algorithms via Randomized Dual Fitting}
\institute{$^{\dag}$Cornell University, Department of Computer Science.}
\toctitle{A Unified Approach to Online Allocation Algorithms via Randomized Dual Fitting}
\maketitle
\begin{abstract}
We present a unified framework for designing and analyzing algorithms for online budgeted allocation problems (including online matching) and their generalization, the Online Generalized Assignment Problem (OnGAP). These problems have been intensively studied as models of how to allocate impressions for online advertising. In contrast to previous analyses of online budgeted allocation algorithms (the so-called ``balance'' or ``water-filling'' family of algorithms) our analysis is based on the method of randomized dual fitting, analogous to the recent analysis of the RANKING algorithm for online matching due to Devanur et al. Our main contribution is thus to provide a unified method of proof that simultaneously derives the optimal competitive ratio bounds for online matching and online fractional budgeted allocation. The same method of proof also supplies $(1-1/e)$ competitive ratio bounds for greedy algorithms for both problems, in the random order arrival model; this simplifies existing analyses of greedy online allocation algorithms with random order of arrivals, while also strengthening them to apply to a larger family of greedy algorithms. Finally, for the more general OnGAP problem, we show that no algorithm can be constant-competitive; instead we present an algorithm whose competitive ratio depends logarithmically on a certain parameter of the problem instance, and we show that this dependence cannot be improved.

\end{abstract}

\section{Introduction}
\label{sec-intro}

Online allocation problems, in which items arriving sequentially must be allocated (either integrally or fractionally) to capacitated servers at the time of their arrival, are one of the most long-standing areas of investigation in online algorithms. In the past decade they have gained a renewed importance as models of how to allocate impressions for online advertising, beginning with the seminal work of Mehta et al. on the AdWords problem~\citep{MSVV}. As results on different versions of online matching and AdWords have proliferated, several recurring themes have become apparent, among them the vital role of primal-dual methods and the omnipresence of the constant $1-1/e$. Despite such striking similarities, the original methods used to prove these results were surprisingly disparate. In this paper we capitalize on a new randomized dual-fitting method introduced by Devanur et al.~\citep{DJK} to show how all of these results can be derived using a common method of proof.\footnote{The final section of~\citep{DJK} sketches derivations of some of these results without providing the full details of the proofs.} In so doing, we also strengthen and generalize some of the existing results.

Two families of algorithms predominate the literature on online allocation problems: the RANKING algorithms originating in the work of Karp, Vazirani, and Vazirani~\citep{KVV}, and the BALANCE algorithms (also known as water-filling) originating in the work of Kalyanasundaram and Pruhs~\citep{KP}. The former family of algorithms associates a random priority to each vertex on the offline side of the bipartite graph that models the potential assignments. Arriving items on the online side are then alocated greedily to the highest-priority neighbor with available capacity. The latter family of algorithms visualizes the fractional allocation process as a process in which each arriving node fills water into buckets corresponding to its neighbors, giving priority to the buckets with the lowest water levels so as to mqximize the minimum neighboring water level when the fractional assignment is completed. Both families of algorithms are known to achieve competitie ratio $1-1/e$ for the respective problems to which they are applied. The original analysis of RANKING~\citep{KVV} and subsequent simplifications~\citep{AGKM11,BM08,GM08} were based on combinatorial and probabilistic arguments. The BALANCE algorithm and its generalization to the AdWords problem~\citep{MSVV} were originally analyzed by direct manipulation of inequalities relating the algorithm's solution to the optimal one. A simplified analysis of BALANCE (and its fractional counterpart, waterfilling) using the primal-dual framework was subsequently presented in~\citep{BJN}.

A related chain of results, also dating back to the seminal work of 
Karp, Vazirani, and Vazirani~\citep{KVV}, concerns the performance of na\"{i}ve 
greedy algorithms for these problems, when the input sequence is presented
in uniformly random order. (We henceforth refer to this as the \emph{random
order arrival} model.) They showed that the greedy algorithm for online 
matching is $(1-1/e)$-competitive in the random order arrival model\footnote{%
It is important to mention that in our competitive analysis, in the worst-case model the competitive ratio is defined to be $\underset{\mathscr{I}}{\min}\frac{\textbf{ALG}(\mathscr{I})}{\textbf{OPT}(\mathscr{I})}$ in which $\text{OPT}$ denotes the optimal offline algorithm. However, in the the random order arrival, the performance of the online algorithm is measured against the worst-case input with random order of arrival on the online side.}.
Goel and Mehta~\citep{GM08} presented a simplified proof of this result
while correcting a technical mistake, and they also proved that a
greedy algorithm for online budgeted allocation is $(1-1/e)$-competitive
in the random order arrival model,
when bids are much less than budgets
(henceforth, the \emph{infinitesimal-bid limit}), 
using a more complicated (and still
essentially combinatorial) analysis. 

Recent work by Devanur et al.~\citep{DJK} points toward the possibility of
unifying these proofs under a common primal-dual framework. They provided 
a primal-dual analysis of the RANKING algorithm for online matching, that
differs from the standard 
primal-dual analysis of the waterfilling algorithm (e.g.,~\citep{BJN})
in at least
two key respects: first, it uses the method of dual fitting 
(meaning that it constructs a dual solution whose value is exactly
equal to the primal one, then rescales it to achieve dual
feasibility); second, it constructs the dual solution via a
randomized procedure that only achieves feasibility in 
expectation (meaning that the expectation of the constructed
dual vector is feasible for the dual LP, but a given execution
of the randomized algorithm may fail to construct a feasible dual).
There are striking similarities between the primal-dual 
analysis of waterfilling in~\citep{BJN} and the primal-dual analysis
of RANKING in~\citep{DJK}; most notably, 
a function $g : [0,1] \to [0,1]$ satisfying the integral
equation $\int_{0}^{\theta} g(y) \, dy + [1 - g(\theta)] = 1 - 1/e$
subject to the boundary condition $g(1)=1$ arises in both.
The final section of~\citep{DJK} sketches an analysis of
a primal-dual algorithm for online budgeted allocation that partly bridges
the gap between the methods of proof in~\citep{BJN}
and~\citep{DJK} by using the dual-fitting method
to establish a competitive ratio of $1-1/e$ 
in the infinitesimal-bid limit.
It also reports, without proof, an observation due to
Aman Dhesi~\citep{AD} that a similar method of proof
can be applied to rederive the theorem of~\citep{GM08}
that the greedy algorithm for online budgeted allocation
is $(1-1/e)$-competitive in the infinitesimal-bid limit,
assuming random order of arrival.

\paragraph{Our contributions.}
Our main contribution in this paper is a unified method
of analysis for deriving and extending these results, using
the randomized dual-fitting method introduced in~\citep{DJK}.
We start by presenting these ideas in the context of online
matching (both integral and fractional versions).
Section~\ref{sec-onlinematching} rederives the $1-1/e$
competitive ratio for the RANKING and waterfilling 
algorithms using essentially the same proof for both
results. It also proves a competitive ratio of $1-1/e$,
in the random order arrival model,
for a class of online fractional matching algorithms
that we call \emph{greedy allocation-monotone algorithms},
again using the randomized dual-fitting technique.
The standard greedy algorithm for online (integral) matching
is a special case of a greedy allocation-monotone fractional
matching algorithm, and thus our proof rederives the 
result of~\citep{GM08,KVV} that Greedy is $(1-1/e)$-competitive
for online matching in the random order arrival model, 
while generalizing the competitive ratio bound to a broader
class of greedy algorithms. (Essentially, unlike the greedy
algorithms analyzed by~\citep{GM08,KVV}, our generalization 
applies to greedy
algorithms in which different vertices on the online side
may use different ordering of the offline side to break ties 
among neighbors. It does not appear that the techniques
of~\citep{GM08,KVV} can be used to analyze such algorithms
in the random order arrival model.) Alongside these results
we also present an improved analysis (again using the primal-dual
method) of the water-filling algorithm for online fractional
matching in degree-bounded graphs, which shows that its
competitive ratio exceeds $1-1/e$ by $\Omega(1/d)$ when
the vertices on the online side have degree bounded by $d$.

In Section~\ref{sec-budg-alloc} we extend these results 
(excluding the bounded-degree one) to 
the online budgeted allocation problem. First we present
a \emph{virtual water-filling} algorithm that is a fractional
counterpart to the AdWords algorithms of~\citep{BJN,MSVV}.
Like~\citep{BJN} we analyze this algorithm using the primal-dual
method, but we depart from that paper in providing a 
randomized dual-fitting analysis that directly generalizes our
analysis of the waterfilling algorithm in Section~\ref{sec-onlinematching}.
We apply the same proof technique yet again to show
that any greedy allocation-monotone
algorithm for online fractional budgeted allocation is
$(1-1/e)$-competitive in the random order arrival model.
An important observation is that, while fractional budgeted
allocation and integral budgeted allocation are essentially
equivalent in the infinitesimal-bid limit under worst-case
arrival order, the two problems are inequivalent in the
random order arrival model. This is because in the random
order arrival model of the fractional allocation problem,
non-infinitesimal bids show up in random order but each of them must be
allocated all at once, at the time of its arrival. When we
replace each non-infinitesimal bid with a multiset of 
infinitesimal copies and permute that multiset into random
order, the event that all of the copies of one bid arrive
consecutively has vanishingly small probability. Thus, our
result for online fractional budgeted allocation in the
random order arrival model has no counterpart in the prior 
literature.

Finally, we consider the online generalized assignment problem
(OnGAP) which generalizes online budgeted allocation by allowing
the \emph{bid value} $b_{i,j}$ for assigning item $i$ to bidder $j$ 
(the amount that it constributes to the objective function)
to differ from its \emph{weight} $w_{i,j}$ (the amount of 
bidder $i$'s budget that it consumes). Under the assumption
that the ratio $b_{i,j}/w_{i,j}$ is always in the range 
$[1,\eta)$ for some $\eta > 1$, we provide an algorithm
with competitive ratio $\Omega(\frac{1}{\log (\eta)})$
for fractional OnGAP. The algorithm works by reducing the
problem to online budgeted allocation, while losing a 
logarithmic factor in the reduction. We also provide a corresponding
impossibility result: the competitive ratio of any fractional
OnGAP algorithm is bounded above by $O(\frac{1}{\log(\eta)})$
in the worst case. (The same problem in a \emph{free disposal}
model was considered by~\citep{FKMMP}, who provided an algorithm
with competitive ratio $1-1/e$ in that model.)

\section{Randomized dual fitting approach to online bipartite matching}
\label{sec-onlinematching}

We start developing our unified \textit{dual fitting} primal-dual framework
for design and analysis of online allocation algorithms by investigating the central problem of online fractional matching. To elaborate on our technique, we consider two models. The first model is the worst-case, in which we provide a simple dual fitting randomized primal-dual analysis of the \textit{water-filling} algorithm, a.k.a BALANCE~\citep{KP}. This is the first analysis for this algorithm that is exactly analogous to the randomized primal-dual analysis of RANKING in~\citep{DJK}. We then expand our proof technique to the random order arrival model, in which we provide a simple dual fitting randomized primal-dual analysis of \textit{allocation-monotone greedy} algorithms (which are defined later). This in turn  rederives a stronger version of the result of~\citep{GM08} on the competitive ratio of the greedy algorithm in the integral online bipartite matching problem under random order arrival.

Our contribution in this section is twofold. On the one hand our techniques are based on simpler primal-dual proofs, and on the other hand they shed some insight on how to unify the analyses of all different online algorithms for the online bipartite matching problem under different models.
\subsection{Notations}
\label{secnot-onlinematching}
Suppose we have a bipartite graph $G=(L\cup R,E)$. In the online bipartite matching problem, $L$ is the \textit{offline} set of vertices known at the beginning and $R$ is the \textit{online} set of vertices. Upon arrival of each vertex $j\in R$ the set of offline neighbors of $j$, denoted by $N(j)$, is revealed to the algorithm. The algorithm can then match $j$ to its neighbors. Fractional solutions are allowed, denoted by $\{x_{i,j}\}$, and the objective is to maximize the value of the algorithm's fractional matching. The following is the standard linear programming formulation for this problem and its dual.
\begin{align}
&\text{maximize~~} \quad \sum_{(i,j)\in E}{x_{i,j}}~~~\text{s.t.}\quad
&\text{minimize~~} \sum_{i\in L }{\alpha_i}+\sum_{j\in R}{\beta_j}~~~\text{s.t.} \nonumber\\
&\sum_{j\in N(i)}{x_{i,j}}\leq 1 , \quad i\in L\quad
&\alpha_i+\beta_j\geq 1\quad (i,j) \in E \nonumber \\
&\sum_{i\in N(j)}{x_{i,j}}\leq 1, \quad j\in R\quad 
&\alpha_i\geq 0,\quad i\in L \nonumber\\
&x_{i,j} \geq 0, \quad (i,j)\in E \quad 
&\beta_j\geq 0,  \quad j\in R \label{lp-matching}
\end{align}
Furthermore, let $y_i\triangleq \sum_{j\in N(i)}{x_{i,j}}$ denote the ``water-level" of vertex $i$. These terms reflect an analogy between the fractional matching and spreading units of water among neighboring buckets of unit capacity. Let us also define $\yinew$ and $\yiold$ to be the water level of $i$ after processing and before processing a vertex $j \in R$, respectively (The identity of the relevant vertex $j$ will always be clear from context, so our notation $\yinew, \yiold$ does not directly refer to $j$).

In our competitive analysis, we use the worst-case model and the random order arrival model. In the worst-case model the competitive ratio is 
measured against the oblivious non-adaptive adversary (worst-case input), while in the the random order arrival, the performance of the online algorithm is measured against the worst-case input with the random order of arrival on the online side, i.e. we assume that vertices in $R$ arrive based on a random permutation sampled from the uniform distribution. 
\subsection{Our framework}
\label{sec-framework}
In an abstract sense, the dual fitting framework can be described in the following way. The goal in this framework is to construct a dual solution accompanying the primal algorithm's solution, such that it satisfies the following two properties.
\begin{itemize}
\item \textbf{ Property (1)}: The dual objective value is equal to the primal objective value.
\item\textbf{ Property (2)}: The dual becomes feasible in expectation if it is divided by $F$, where $F\in [0,1]$ is a constant.
\end{itemize}
If we can successfully construct such a dual solution, then we can divide the dual by $F$  to make it feasible in expectation, i.e. $\alpha_i\leftarrow \alpha_i/F,\beta_j\leftarrow \beta_j /F,\forall i,j $.  The lower bound $F$ for the competitive ratio of the algorithm is then guaranteed immediately comparing the objective value of the primal algorithm with the objective value of the expectation of the dual solution (which is feasible for the dual LP). 

To make a dual solution satisfying the first property, once vertex $j$ arrives we update $\beta_j$ in the following way.
\begin{equation}
\beta_j=(\Delta\text{primal})_j-\sum_{i\in N(j)}\Delta \alpha_i
\end{equation}
in which $\Delta \alpha_i$ is the change in $\alpha_i$ and $(\Delta \text{primal})_j$ is the change in the primal objective value during processing $j$. The update rule of $\alpha_i$ depends on the setting, and usually is relying on a function $g:{[0,1]\rightarrow [0,1]}$, which is non-decreasing and satisfies $g(1)=1$ along with a certain integral equation, e.g.\ for the online bipartite matching we need 
\begin{equation}
\label{inteq1}
\forall t\in [0,1]:\int_0^{t}g(x)dx +1-g(t) =F
\end{equation} 
 Finally, in some problems it might be advantageous to define some \textit{auxiliary} random variables to facilitate the dual construction. These random variables help us to keep track of the dual variables $\{\alpha_i\}$.

In subsequent sections, we see  instances of this framework in detail, as mentioned before.

\subsection{Dual fitting analysis for the worst-case model}
We first consider the worst case model, and analyze the water-filling algorithm, introduced first in~\citep{KP}, using our framework. 
\subsubsection{Water-filling algorithm} The algorithm is very simple. For each arrived vertex $j$, while $\{\exists$ neighbor $i\in N(j)$ with ${y_i}<1\}$ and $\{\sum_{i\in N(j)}{x_{i,j}}<1\}$ the algorithm allocates vertex $j$ \textit{continuously} among the neighbors with the minimum water-level, so as to increase their water-levels equally. We now have a theorem on the competitive ratio of water-filling algorithm\footnote{Following the upper bound of $\frac{e-1}{e}$ for the worst-case competitive ratio of online fractional matching \citep{KVV}, this bound is tight.}.
\begin{theorem}\label{th-worstcase-water}\citep{KP}
In the fractional online bipartite matching problem, the water-filling algorithm achieves the competitive ratio of $\frac{e-1}{e}$ under the worst-case model.
\end{theorem}
\subsubsection{Proof of theorem~\ref{th-worstcase-water} by dual fitting}
\begin{proof}
Consider the primal-dual linear programs in (\ref{lp-matching}), and let $\{x_{i,j}\}$ be the primal algorithm's allocation. Prior to the running of the algorithm, generate i.i.d. random variables $\{U_i\}\sim\text{unif}[0,1]$ for each $i\in L$. We now construct a randomized dual solution, as prescribed by the dual fitting framework. For each vertex $i\in L$,  initialize $\alpha_i=\beta_j=0,\forall i,j $. Once vertex $j\in R$ arrives,  let $y_{i,\text{old}}$ and $y_{i,\text{new}}$ denote the water level of $i\in N(j)$ before and after processing the arrived vertex by the primal algorithm. Then update the dual variables by setting
\begin{equation}
\alpha_i=\left\{ \begin{array}{ll}
g(U_i) &\mbox{ if $U_i\leq y_{i,\text{new}}$} \\
0 &\mbox{ if $U_i > y_{i,\text{new}}$}
\end{array} \right.,~ \forall i \in N(j),
~~~~\beta_j=(\Delta \text{primal})_j-\sum_{i\in N(j)}{\Delta \alpha_i}=\sum_{i\in N(j)}({x_{i,j}}-\Delta \alpha_i)
\end{equation}
in which, $g$ is a non-decreasing function satisfying $g(1)=1$ (we use this property later in the proof) and $\Delta \alpha_i$ is the change in $\alpha_i$ during processing $j$. Let $(\Delta\text{dual})_j$ and  $(\Delta\text{primal})_j$ denote the change in the dual and primal objective during the time algorithm processes $j$. Clearly we have $(\Delta\text{dual})_j=\sum_{i\in N(j)}{x_{i,j}}=(\Delta\text{primal})_j$, and so the first property holds. To show the second property, as the first step notice that $\mathbb{E}\{\alpha_i\}\geq 0, \forall i\in L$. As we show later, $\mathbb{E}\{\beta_j\}\geq 0, j\in R$. Now,  fix an edge $(i,j)\in E$. To proceed to the next step, we need the following definition.
\begin{definition}
For a fixed edge $(i,j)\in E$,  the \textit{critical water-level}, denoted by $Y^c$, is  the final water-level of $i$ after the algorithm finishes processing $j$. 
\end{definition}
Let $L(j)$ denotes the subset of $N(j)$ whose water level has been increased during processing $j$.  We now prove the following lemmas. (The names of the lemmas
are chosen to match the analogous lemmas from the primal-dual analysis of the RANKING algorithm in~\citep{DJK}.)
\begin{lemma}[Dominance Property]
\label{dom-matching} For a fixed edge $(i,j)\in E$, $\mathbb{E}\{\alpha_i\}\geq \int_{0}^{Y^c} g(u) du$.
\end{lemma}
\begin{proof}
From the definition of $\alpha_i$, if the final water-level of $i$ after termination is $\theta$, then $\mathbb{E}\{\alpha_i\}=\int_{0}^{\theta} g(u) du\geq \int_{0}^{Y^c} g(u) du$.\qed
\end{proof}
\begin{lemma}[Monotonicity Property] For a fixed edge $(i,j)\in E$,
$\mathbb{E}\{\beta_j\} \geq 1-g(Y^c).$
\end{lemma}
\begin{proof}
We have 
\begin{align}
\mathbb{E}\{\beta_j\}&= \sum_{i'\in N(j)}{x_{i',j}}-\sum_{i'\in L(j)}{\mathbb{E}\{\Delta \alpha_{i'}\}}= \sum_{i'\in L(j)}{x_{i',j}}-\sum_{i'\in L(j)}{\int_{y_{i',\text{old}}}^{y_{i',\text{new}}}g(u)du}\nonumber\\
&\overset{(1)}{\geq} \sum_{i'\in L(j)}{x_{i',j}}-\sum_{i'\in L(j)}{g(y_{i',\text{new}})(y_{i',\text{new}}-y_{i',\text{old}})}\nonumber\\
&\overset{(2)}{=}\sum_{i\in L(j)}{x_{i',j}}(1-g(y_{i',\text{new}}))\overset{(3)}{\geq} (1-g(Y^c))\sum_{i'\in L(j)}{x_{i',j}}~(\geq 0)\label{worst-mono}
\end{align}
where $(1)$ is true as $g$ is non-decreasing,  $(2)$ holds because $y_{i',\text{new}}-y_{i',\text{old}}=x_{i',j}$, and $(3)$ is true because in the water-filling algorithm $\forall i'\in L(j), y_{i',\text{new}}\leq y_{i,\text{new}}=Y^c$. There are two cases to consider. The first case is when upon the completion of processing  $j$, $\sum_{i'\in N(j)}x_{i',j}< 1$, which implies $Y^c=1$. The fact that $\mathbb{E}\{\beta_j\} \geq 1 - g(Y^c)$ is now immediate, as $\mathbb{E}\{\beta_j\} \geq 0=1-g(1)$. (Recall that $g(1)=1$.) The second case is when $\sum_{i'\in N(j)}x_{i',j}=1$, in which case the proof is again immediate following (\ref{worst-mono}). \qed
\end{proof}
The monotonicity property implies also $\mathbb{E}\{\beta_j\} \geq 0$ as a side result. Now, for any $g$ (non-decreasing, $g(1)=1$) and $F$ that satisfies (\ref{inteq1}), the fact that the dual becomes feasible in expectation upon dividing by $F$ is proved by combining the above two lemmas: $$\mathbb{E}\{\alpha_i+\beta_j\}\geq \int_{0}^{Y^c} g(u) du+1-g(Y^c)=F.$$ The proof of Theorem~\ref{th-worstcase-water} then follows by setting $g(x)=e^{x-1}$ and $F=\int_0^{1}g(x)dx=\frac{e-1}{e}.$ 
\qed 
\end{proof}
\begin{remark}
Surprisingly, the proof of Theorem~\ref{th-worstcase-water} also derives an analysis of RANKING for the integral online bipartite matching problem. In fact, an interested reader can check that for this special case, in which $y_i\in\{0,1\}$, our update rule boils down to the update rule of dual variables in \citep{DJK}, and the rest of the proof follows the same steps as in the proof of Theorem~4 in \citep{DJK}. Moreover, the same proof technique can be applied to the vertex weighted matching problem to re-derive the proof of \citep{DJK} for the result of \citep{AGKM11}. 
\end{remark}
\subsubsection{Water-filling in bounded degree online bipartite matching}
In the appendix, we provide an improved result for the competitive ratio of  water-filling when the underlying graph has bounded degree $d$ on the online side. The new analysis is basically very similar to our dual-fitting framework, with some steps improved due to the bounded-degree condition. We have the following result whose proof is deferred for space reasons.
\begin{theorem}
\label{th4}
The water-filling algorithm achieves a competitive ratio of $\frac{e-1}{e}(1+c/d)$ for some constant $c>0$ in the $d-$bounded degree fractional online bipartite matching problem.
\end{theorem}

\subsection{Dual fitting analysis for the random order arrival model }
\label{secroa}
In this section, we first define a broad class of fractional matching algorithms called \textit{allocation-monotone greedy} algorithms, and then we analyze their competitive ratio under random order arrival using our dual fitting framework. The integral greedy algorithm with a fixed ordering on the offline side is a special case of our broad class. So, the result of this section also derives the result of \citep{GM08,KVV} as a special case.

We focus on algorithms satisfying two natural properties formalized by
the following two definitions.
\begin{definition} \label{def:greedy:onlinematching}
An online fractional matching algorithm is \emph{greedy} if it is not allowed to keep a positive fraction of the arrived vertex unmatched whenever fractional matching is possible.
\end{definition}

To define allocation monotonicity, we need the following notations. Fix the offline side $L$. Let an input sequence $\textbf{seq}$ be an ordered set of online vertices. (The data associated to each element of the ordered set encodes its set of offline neighbors as well.) For any online fractional matching algorithm, let $\mathbf{y}_{\textbf{old}}(\textbf{seq})$ and $\mathbf{y}_{\textbf{new}}(\textbf{seq})$ denote the vector of water-levels of vertices in $L$ when running the algorithm on the input sequence $[\textbf{seq},j]$, before and after algorithm processes $j$. Here, the notation $[\textbf{seq},j]$ denotes the concatenation of $\textbf{seq}$ with the single-element sequence $(j)$.
\begin{definition} 
\label{defall-onlinematching}
An online fractional matching algorithm is \emph{allocation-monotone} if for any two input sequences $[\textbf{seq}_1,j]$ and $[\textbf{seq}_2,j]$ such that $\mathbf{y}_{\textbf{old}}(\textbf{seq}_1)\preceq\mathbf{y}_{\textbf{old}}(\textbf{seq}_2)$, then $\mathbf{y}_{\textbf{new}}(\textbf{seq}_1)\preceq\mathbf{y}_{\textbf{new}}(\textbf{seq}_2)$. (Here, $\preceq$ denotes the element-wise vector comparison.)
\end{definition}
 
 Now, we have the following lower bound on the competitive ratio of an \textit{arbitrary} allocation-monotone greedy algorithm under the random order arrival model in the online bipartite matching problem.
\begin{theorem}
\label{th-roa-onlinematching}
In the online bipartite matching problem, any allocation-monotone greedy algorithm achieves a competitive ratio of at least $\frac{e-1}{e}$ under the random order arrival model.
\end{theorem}
Before proving Theorem~\ref{th-roa-onlinematching}, we note
the following corollary. 
\begin{corollary}
\label{greedy:integralmatching}
In the integral online bipartite matching problem, an integral greedy matching algorithm with a fixed ordering on the offline side is indeed a special case of an allocation-monotone greedy algorithm, and hence achieves the competitive ratio of $\frac{e-1}{e}$ under the random order arrival model. 
\end{corollary}
\begin{remark}
Corollary~\ref{greedy:integralmatching}, which is also proved in  \citep{GM08}, is a simple corollary of the analysis of RANKING in \citep{KVV} together with the observation that RANKING and greedy with random order of arrival are dual to each other~\citep{KVV,MY11}. Interestingly, we can push this result one step further and say that greedy with different fixed orderings on the offline side for each arriving vertex achieves the same competitive ratio, as it is still allocation-monotone, although there is no longer any duality principle relating this family of greedy matching algorithms to the RANKING algorithm.
\end{remark}
\begin{proof}[of Theorem~\ref{th-roa-onlinematching}]
Consider the primal-dual linear programs in (\ref{lp-matching}), and let $\{x_{i,j}\}$ be the primal algorithm's allocation. In order to simulate the uniform random order of arrival of vertices, instead of picking a random total ordering on vertices in $R$, we generate i.i.d. random variables $\{Z_j\} \sim \text{unif} [0,1]$  for every $j$ in $R$ prior to the running of the algorithm, and then assume vertices in $R$ arrive based on the sorted order of $Z_j$'s. (The lower the value of $Z_j$, the earlier $j$ arrives.) We now construct a randomized dual solution as following. Initialize $\alpha_i=\beta_j=0$ for all $i,j$. Once vertex $j\in R$ arrives, for every $i\in N(j)$ let
 \begin{equation} 
 \alpha_i^{\text{new}}=\alpha^{\text{old}}_i+x_{i,j}(1-g(Z_j)),~~~~ \beta_j =(\Delta\text{primal})_j-\sum_{i\in N(j)}\Delta \alpha_i=
\left(\sum_{i\in N(j)}{x_{i,j}} \right)g(Z_j).
 \end{equation} 
Here, $g:{[0,1]\rightarrow[0,1]}$ is a non-decreasing function satisfying $g(1)=1$ (which will be used later in the proof). Obviously, during processing every arrived vertex $j\in R$, 
$$(\Delta{\text{dual}})_j=\beta_j+\sum_{i\in N(j)}{\Delta \alpha_i}=\sum_{i\in N(j)}{x_{i,j}}=(\Delta {\text{primal}})_j$$ 
and hence property (1) is satisfied.

 In order to show that the dual solution satisfies property (2), consider a fixed $(i,j)\in E$. Clearly, $\beta_j\geq 0, \forall j\in R$ and $\alpha_i\geq 0, \forall i\in L$. Now, consider an instance of the algorithm run on the graph $G\backslash \{j\}$ with the same choice of ${Z}_{j'}$ for all vertices in $j'\in L\backslash \{j\}$. We need the following definition.
\begin{definition}
\label{def:critical}
For a fixed $(i,j)\in E$ and fixed values of $\mathbf{Z_{-j}}$, the ``critical position function" of $i$, denoted by $Z^c:{[0,1]\rightarrow \{\mathbf{Z}_{-j}\}}$, is the $Z$-value of the arrived vertex allocating to $i$ at the time that the water level of $i$ reaches the value $y$ while running the algorithm on $G\backslash\{j\}$. 
\end{definition}
Clearly, $Z^c$ is a non-decreasing step-function by definition. Suppose $Z^c$ has steps at $0=\theta_0\leq\theta_1\leq\ldots\leq\theta_r(=\theta)\leq \theta_{r+1}=1$, in which $\theta$ is the final water level of $i$ at termination. Moreover, suppose
$\forall {y}\in[\theta_{k-1},\theta_{k}), Z^c({y})=Z_{a_k},$
in which $Z_{a_1}\leq\ldots\leq Z_{a_r}$ is a non-decreasing sequence in $\{\mathbf{Z}_{-j}\}$. As a convention we assume $Z^c({y})=1, {y}\in(\theta,1]$, or equivalently $Z_{a_{r+1}}=1$. We now prove  two lemmas.
\begin{lemma}[Dominance Property]\label{lemma1-onlinematching}
For fixed values of $\mathbf{Z}_{-j}$, 
\begin{equation}
\mathbb{E}\{\beta_j| \mathbf{Z}_{-j}\}\geq \int_{0}^{1}g(z)dz-\sum_{k=1}^{r}\theta_k \int_{Z_{a_k}}^{Z_{a_{k+1}}}g(z)dz
\end{equation}
\end{lemma}
\begin{proof}
If $Z_j\in[Z_{a_k},Z_{a_{k+1}})$ for some $k$, then water level of $i$ at the time that $j$ arrives is $\theta_k$. Hence, following the fact that the algorithm is greedy, $\Delta \text {primal}$ will be at least $1-\theta_k$.
So,
\begin{align}
&\label{dom-matching-1}\mathbb{E}\{\beta_j \mathds{1}(Z_j\in [0,Z_{a_{1}})| \mathbf{Z}_{-j}\}=\int_{0}^{Z_{a_1}}g(z)dz,\\
&\label{dom-matching-2}\forall k\leq r:\mathbb{E}\{\beta_j \mathds{1}(Z_j\in [Z_{a_{k}},Z_{a_{k+1}})| \mathbf{Z}_{-j}\}\geq(1-\theta_k)\int_{Z_{a_k}}^{Z_{a_{k+1}}}g(z)dz.
\end{align}
Since $\mathbb{E}\{\beta_j| \mathbf{Z}_{-j}\}=\sum_{k=0}^{r} {\mathbb{E}\{\beta_j \mathds{1}(Z_j\in [Z_{a_{k}},Z_{a_{k+1}}))| \mathbf{Z}_{-j}\}}$, the lemma follows by summing up the right hand sides of (\ref{4-1}) and (\ref{4-2}) for $1\leq k\leq r.$\qed
\end{proof}
\begin{lemma} [Monotonicity Property] \label{lemma2-onlinematching} For fixed values of $\mathbf{Z}$, 
\begin{equation}
\alpha_i\geq \int_{0}^{\theta} 1-g(Z^c({y})) \, d{y}~.
\end{equation}
\end{lemma}
\begin{proof}
Looking at the instance of the algorithm on $G\backslash \{j\}$, $\frac{d\alpha_i}{d{y}_i}(y_i)$, i.e. the rate at which $\alpha_i$ increases w.r.t. the water-level of $i$, is equal to $1-g(Z^c({y}_i))$ by the definition of $Z^c({y}_i)$. (We need $g(1)=1$ to ensure $\frac{d\alpha_i}{d{y}_i}=0$ whenever ${y}_i>\theta$.) Suppose $Z^c(\theta)=Z_{j'}$ for some $j'\in R\backslash\{j\}$. Now,  insert $j$ back to the graph and run the algorithm again. If $Z_j>Z_{j'}$, then nothing changes until arrival of $j'$, and hence $\frac{d\alpha_i}{d{y}_i}({y}_i)= 1-g(Z^c({y}_i))$ for ${y}_i\in[0,\theta]$.  If $Z_j<Z_{j'}$, suppose ${y}_i=\theta^{c}$ at the time that $j$ arrives. Obviously,  $\frac{d\alpha_i}{d{y}_i}({y}_i)= 1-g(Z^c({y}_i))$ for ${y}_i\in[0,\theta^{c}]$. For ${y}_i\in[\theta^c,\theta]$, suppose $j''\in R$ is the vertex that pours water into $i$ at the time that it reaches the water-level ${y}_i$.  Now, compare the water-level of neighbors of $j''$ before and after inserting $j$. The  effect of inserting vertex $j$ back into the graph is just increasing the water-level of a subset of its neighbors. Now, by a simple induction on the time, using the allocation-monotonicity property one can conclude that the water-level of neighbors of vertex $j''$ at the time of its arrival when $j$ is in the graph is at least as when $j$ was not in the graph. Hence, $Z^c({y}_i)\geq Z_{j''}$, which then implies $\frac{d\alpha_i}{d{y}_i}({y}_i)=1-g(Z_{j''})\geq 1-g(Z^c({y}_i))$, as $g$ is a non-increasing function of its argument. Note also that the final water-level of $i$ when we insert back $j$ to the graph, denoted by $y_{i,\text{final}}$, is at least $\theta$, again due to allocation-monotonicity of the algorithm. The lemma is then proved immediately, as $\alpha_i=\int_{0}^{{y}_{i,{\text{final}}}}\frac{d\alpha_i}{d{y}_i} \, d{y}_i \geq \int_{0}^{\theta} 1-g(Z^c({y}_i)) \, d{y}_i$.\qed
\end{proof}
Now, suppose $g$ (non-decreasing, with $g(1)=1$) and $F$ satisfy the integral equation in (\ref{inteq1}). 
\begin{align}
\mathbb{E}\{\beta_j+&\alpha_i|\mathbf{Z}_{-j}\}\geq \int_{0}^{1}g(z) \, dz-\sum_{k=1}^{r}\theta_k \int_{Z_{a_k}}^{Z_{a_{k+1}}}g(z)dz+\int_{0}^{\theta} 1-g(Z^c(y)) \, dy\nonumber\\
&=F-\sum_{k=1}^{r}\theta_k \int_{Z_{a_k}}^{Z_{a_{k+1}}}g'(z) \, dz+\sum_{k=1}^{r}(\theta_k-\theta_{k-1})(1-g(Z_{a_k}))\nonumber\\
&=F-\sum_{k=1}^{r}\theta_k(g(Z_{a_{k+1}})-g(Z_{a_{k}}))+\sum_{k=1}^{r+1}(\theta_k-\theta_{k-1})(1-g(Z_{a_k}))\nonumber\\
&=F-\sum_{k=1}^{r}\theta_k(g(Z_{a_{k+1}})-g(Z_{a_{k}}))+\sum_{k=1}^{r}\theta_k(g(Z_{a_{k+1}})-g(Z_{a_{k}}))=F
\end{align}
in which the first inequality follows from lemmas (\ref{lemma1-onlinematching}) and (\ref{lemma2-onlinematching}), second equality follows from $g'(z)=g(z)$ (due to equation (\ref{inteq1})) and from the definition of $Z_{a_k}$, third equality follows 
by adding $(\theta_{r+1}-\theta_r)(1-g(Z_{a_{r+1}}))=0$ (as $g(1)=1$) to the RHS, and the final equality comes from rearranging the terms of the last telescopic sum. Hence, $\mathbb{E}\{\alpha_i+\beta_j\}\geq F$, and property (2) holds. The proof is completed by plugging in $g(x)=e^{x-1}$ and $F=\int_0^{1}g(x)dx=\frac{e-1}{e}$\qed
\end{proof}

\section{Online budgeted allocation}
\label{sec-budg-alloc}
In this section, we extend our randomized dual-fitting framework to 
the Online Budgeted Allocation Problem (OnBAP), a.k.a.\ AdWords.
This general problem includes online bipartite matching and online $b-$matching~\citep{KP} as special cases. Moreover, the fractional version of this problem generalizes the fractional vertex weighted online bipartite matching~\citep{AGKM11}.

We again focus on two different settings. The first setting is the worst-case model. Under this setting, \citep{MSVV} proposed a $\frac{e-1}{e}$ competitive algorithm for the integral OnBAP in the infinitesimal-bid limit. We analyze the fractional version of this algorithm, named the \textit{virtual water-filling algorithm}, using our technique. Our proof is a generalization of the proof of Theorem~\ref{th-worstcase-water}. The second setting is the random order arrival model, and analyzing the greedy algorithm. This problem have been studied extensively in \citep{GM08}, where they have shown in the online integral budgeted allocation problem (in the infinitesimal-bid limit) under the random order arrival model, the greedy algorithm that allocates to the maximum bidder among buyers with enough remaining budget achieves a competitive ratio of $\frac{e-1}{e}$. By extending our definition of allocation-monotone  greedy algorithms to the OnBAP, we derive the same result for this more generalized class of fractional and integral algorithms using our proof technique. Again, our proof is a generalization of the proof of Theorem~\ref{th-roa-onlinematching}. We further show that our results can simply be applied to the integral budgeted allocation problem in the infimitesimal-bid limit to obtain results of \citep{GM08}.\footnote{Interestingly, as discussed in Section~\ref{sec-intro}, unlike in the worst-case model, it is not possible to derive results for fractional OnBAP in the random order arrival model by reducing to the integral problem in the infinitesimal-bid limit.}

\subsection{Notations}
\label{secnot}
An instance of  OnBAP is essentially an instance of online bipartite matching where $L$ is the offline set of buyers with budgets $\{B_i\}$ known at the beginning and $R$ is the set of items arriving online. Moreover, for each arrived item $j$, each bidder $i\in N(j)$ bids a value $b_{i,j}$. No buyer can exceed her budget, and fractional allocations are allowed.  The objective  is to maximize the sum of the allocated bids to  buyers in an online fashion. We also use the following standard linear programming formulation of the OnBAP problem and its dual.
\begin{align}
&\text{maximize~~} \quad \sum_{(i,j)\in E}{b_{i,j}x_{i,j}}~~~\text{s.t.}\quad
&\text{minimize~~} \sum_{i\in L }{B_i\alpha_i}+\sum_{j\in R}{\beta_j}~~~\text{s.t.} \nonumber\\
&\sum_{j\in N(i)}{b_{i,j}x_{i,j}}\leq B_i , \quad i\in L\quad
&b_{i,j}\alpha_i+\beta_j\geq b_{i,j}\quad (i,j) \in E \nonumber \\
&\sum_{i\in N(j)}{x_{i,j}}\leq 1, \quad j\in R\quad 
&\alpha_i\geq 0,\quad i\in L \nonumber\\
&x_{i,j} \geq 0, \quad (i,j)\in E \quad 
&\beta_j\geq 0,  \quad j\in R \label{lp-oba}
\end{align}
Furthermore, let $y_i\triangleq \sum_{j\in N(i)}{b_{i,j}x_{i,j}}$ denote the ``water-level" and let $\bar{y}_i=y_i/B_i$ denote the ``normalized water-level" of vertex $i$. Let us also define $\yinew$ and $\yiold$ to be the water level of $i$ after processing and before processing an item $j \in R$ respectively, as in Section~\ref{secnot-onlinematching}.
\subsection{OnBAP under worst-case model}
In this section we define the virtual water-filling algorithm and characterize the best achievable worst-case competitive ratio in OnBAP using our randomized dual-fitting framework. This in turn sheds some light on how our unified  approach works in different problems.
\subsubsection{The virtual water-filling algorithm}
Let $g:{[0,1]\rightarrow [0,1]}$ be a non-decreasing function, and $\tilde{y_i}={b_{i,j}}(g(\frac{y_i}{B_i})-1)$ denotes the \textit{virtual water-level} of $i\in N(j)$. Then virtual water-filling algorithm is described as follows.
\begin{algorithm}[h]
\caption{ Virtual water-filling algorithm}\label{vwater}
\begin{algorithmic}[1]
   \State \textbf{Initialization$~$} $y_i \gets 0$ for every $i\in L$, $x_{i,j}\gets 0$ for every $(i,j)\in E$.
   \For {arrived $j\in R$} 
  
   \While {$\{\exists$ buyer $i\in N(j)$ with $\tilde{y_i}<0\}$ and $\{\sum_{i\in N(j)}{x_{i,j}}<1\}$}
   \State Allocate item $j$ \textit{continuously} among the buyers with the minimum virtual water-level, so as to increase their virtual water-levels equally. Update $y_i\gets y_i+b_{i,j}x_{i,j}$ and $\tilde{y_i}=b_{i,j}(g(\frac{y_i}{B_i})-1)$ continuously for each $i\in N(j)$ according to the allocation.   
   \EndWhile
   \EndFor
   \State \textbf{Return }{$\{x_{i,j}\}_{(i,j)\in E}$}
\end{algorithmic}
\end{algorithm}

We have a theorem on the competitive ratio of virtual water-filling algorithm\footnote{This bound is again tight, due to its tightness in the special case of online bipartite fractional matching.}.
\begin{theorem}\label{th1}
  In OnBAP under worst-case model, there exists a non-decreasing function $ g:{[0,1]\rightarrow[0,1]}$ such that the corresponding virtual water-filling algorithm is $(\frac{e-1}{e})$-competitive.
\end{theorem}
The proof is similar to the proof of Theorem~\ref{th-worstcase-water}, and hence we just go through it briefly.
\begin{proof}
Consider the primal-dual linear programs in (\ref{lp-oba}), and let $\{x_{i,j}\}$ be the primal algorithm's allocation. Now, follow the exact same steps and notations as in the proof of Theorem~\ref{th-worstcase-water}, and construct a randomized dual solution prescribed by the framework in Section~\ref{sec-framework}. For each vertex $i\in L$,  initialize $\alpha_i=\beta_j=0,\forall i,j $. Upon arrival of $j$, update the dual variables by setting
\begin{equation}
\alpha_i=\left\{ \begin{array}{ll}
g(U_i) &\mbox{ if $U_i\leq \tilde{y}_{i,\text{new}}$} \\
0 &\mbox{ if $U_i > \tilde{y}_{i,\text{new}}$}
\end{array} \right.,~\forall i \in N(j),
~~~\beta_j=(\Delta \text{primal})_j-\sum_{i\in N(j)}{B_i\Delta \alpha_i}
\end{equation}
in which $\tilde{y}_{i,\text{new}}=\frac{y_{i,\text{new}}}{B_i}$. We again have $(\Delta\text{dual})_j=\sum_{i\in N(j)}{b_{i,j}x_{i,j}}=(\Delta\text{primal})_j$, and so the first property holds. Similar to the proof of Theorem~\ref{th-worstcase-water},  $\mathbb{E}\{\alpha_i\}\geq 0, \forall i\in L$. Now,  fix an edge $(i,j)\in E$, and adapt the definition of the critical water-level for OnBAP, i.e. let $\bar{Y}^c$ be the normalized water-level of $i$ after the algorithm finishes processing $j$. We have the following lemmas as before.
\begin{lemma}[Dominance Property] For a fixed $(i,j)\in E$, $\mathbb{E}\{\alpha_i\}\geq \int_{0}^{\bar{Y}^c} g(u) du$.
\end{lemma}
\begin{proof}
By the definition of $\alpha_i$, if the normalized final water-level of $i$ after termination is $\theta$, then $\mathbb{E}\{\alpha_i\}=\int_{0}^{\theta} g(u) du\geq \int_{0}^{\bar{Y}^c} g(u) du$.\qed
\end{proof}
\begin{lemma}[Monotonicity Property] For a fixed $(i,j)\in E$,
$\mathbb{E}\{\beta_j\} \geq b_{i,j}(1-g(\bar{Y}^c))~.$
\end{lemma}
\begin{proof} As in the proof of Lemma~\ref{dom-matching}, we have
\begin{align}
\mathbb{E}\{\beta_j\}&= \sum_{i'\in N(j)}{b_{i',j}x_{i',j}}-\sum_{i'\in L(j)}{\mathbb{E}\{B_{i'}\Delta \alpha_{i'}\}}= \sum_{i'\in L(j)}{b_{i',j}x_{i',j}}-\sum_{i'\in L(j)}{B_{i'}\int_{y_{i',\text{old}}/B_{i'}}^{y_{i',\text{new}}/B_{i'}}g(u)du}\nonumber\\
&\geq \sum_{i'\in L(j)}{b_{i',j}x_{i',j}}-\sum_{i'\in L(j)}{B_{i'}g(y_{i',\text{new}}/B_{i'})(y_{i',\text{new}}-y_{i',\text{old}})/B_{i'}}\nonumber\\
&=\sum_{i'\in L(j)}{x_{i',j}b_{i',j}(1-g(\bar{y}_{i',\text{new}}))}
\overset{(1)}{\geq}b_{i,j}(1-g(\bar{Y}^c))\sum_{i'\in L(j)} x_{i',j}~(\geq 0)\label{worst-mono-2}
\end{align}
in which $(1)$ holds because in the virtual water-filling algorithm, 
$$\forall i' \in L(j): b_{i',j}(g(\bar{y}_{i',\text{new}})-1)\leq b_{i,j}(g(\bar{y}_{i,\text{new}})-1).$$ There are two cases to consider. When $\sum_{i'\in N(j)}x_{i',j}< 1$ holds upon the completion of processing $j$ (which implies $\bar{Y}^c=1$), then $\mathbb{E}\{\beta_j\} \geq 0=1-g(1)$ which completes the proof.  Otherwise $\sum_{i'\in N(j)}x_{i',j}=1$, and the proof follows immediately from (\ref{worst-mono-2}). \qed
\end{proof}
The monotonicity property implies also $\mathbb{E}\{\beta_j\} \geq 0$ as a side result. The proof is then completed by setting $g(x)=e^{x-1}$ and $F=\int_0^{1}g(x)dx=\frac{e-1}{e}$, as in the proof of Theorem~\ref{th-worstcase-water}, and summing the bounds in the two lemmas to conclude that $\mathbb{E}\{b_{i,j}\alpha_i+\beta_j\}\geq Fb_{i,j}$.\qed
\end{proof}
\subsection{OnBAP under the random order arrival model }
\label{secroa}
We now analyze the competitive ratio of greedy algorithms in OnBAP under the random order arrival model using our randomized dual-fitting framework. We focus on a natural class of algorithms that are allocation-monotone and greedy. These two properties were defined in the online matching context in Section~\ref{sec-onlinematching}, and we now generalize those definitions to OnBAP.

 We should assert that our result in this section requires a natural assumption about the input instance. We assume an item can be fully allocated to any individual buyer, i.e. $b_{i,j}\leq B_i, \forall (i,j) \in E$. (Bids are less than or equal to budgets, but not necessarily much less.)

\begin{definition} \label{def:greedy}
In OnBAP, an online fractional allocation algorithm is \emph{greedy} if 
\begin{itemize}
\item[1.]  It is not allowed to keep a positive fraction of the arrived item unallocated whenever allocation is possible.
\item[2.] It always selects the bidder with maximum bid among those neighbors whose budget is not exhausted, with an arbitrary tie-breaking rule.
\end{itemize}
\end{definition}
\begin{remark}
For the special case of online fractional matching, the second part of Definition~\ref{def:greedy} is vacuous since $b_{i,j}= 1$ for all $i,j$, and hence this definition boils down to the definition of greedy algorithms in Section~\ref{sec-onlinematching}.
\end{remark}

In OnBAP, allocation-monotonicity is \emph{exactly} defined as in the online bipartite matching, i.e. Definition~\ref{defall-onlinematching}. (Here, $\textbf{seq}$ is an ordered set of items, such that the data associated with each element of the ordered set encodes its set of offline neighbors along with bids.)

From the above definitions, the following corollary can be seen immediately.
\begin{corollary}
In OnBAP, an algorithm is allocation-monotone greedy if and only if it is greedy and ``breaks ties" (when there is more than one vertex with maximum bid) in an allocation monotone way, as in Definition~\ref{defall-onlinematching}.
\end{corollary}
 We now state a lower bound on the competitive ratio of an \textit{arbitrary} allocation-monotone greedy algorithm for OnBAP under the random order arrival model. 
\begin{theorem}
\label{th3}
In OnBAP under $b_{i,j}\leq B_i, \forall (i,j)\in E$, any allocation-monotone greedy algorithm achieves a competitive ratio of at least $\frac{e-1}{e}$ under the random order arrival model.
\end{theorem}
\begin{remark}
 Following our result for the fractional online budgeted allocation problem, using the standard technique of independent randomized rounding of the fractional solution, one can obtain a $\frac{e-1}{e}$-competitive algorithm for integral OnBAP in the infinitesimal-bid limit. This result has been stated in \citep{GM08}. In Section~\ref{secoba}, we will also provide a simple primal-dual reduction to obtain this result from ours without the need for rounding.
\end{remark}
\begin{proof}
Consider the standard primal-dual LP's for the OnBAP problem and its dual described in (\ref{lp-oba}), and let $\{x_{i,j}\}$ be the feasible fractional allocation of the primal online algorithm. Proof follows exactly the same steps as in the proof of Theorem~\ref{th-roa-onlinematching}, except we update the dual variables as follows for every $i\in N(j)$, upon arrival of $j$.
 \begin{equation} 
 \alpha_i^{\text{new}}=\alpha^{\text{old}}_i+\frac{b_{i,j}x_{i,j}(1-g(Z_j))}{B_i},~~~\beta_j =(\Delta\text{primal})_j-\sum_{i \in N(j)}{B_i\Delta \alpha_i} = \left( \sum_{i\in N(j)}{b_{i,j}x_{i,j}} \right)g(Z_j) 
 \end{equation} 
The function $g$ is as described in the proof of Theorem~\ref{th-roa-onlinematching}. Clearly $\mathbb{E}\{\alpha_i\}\geq 0$ and $\mathbb{E}\{\beta_j\}\geq 0, \forall i,j$. Moreover, during processing every arrived vertex $j\in R$, 
$$(\Delta{\text{dual}})_j=\beta_j+\sum_{i\in N(j)}{B_i\Delta \alpha_i}=(\Delta {\text{primal}})_j$$
and hence property (1)  of Section~\ref{sec-framework} is satisfied by the dual solution. In order to show that property (2) holds,  we follow exactly the same steps as in proof of Theorem~\ref{th-roa-onlinematching} by defining the \emph{critical position function}, i.e.  $Z^c(\bar{y_i})$, as in Definition~\ref{def:critical} by replacing the water-level with normalized water-level. Let also $Z^c$ be the same step function as in the proof of Theorem~\ref{th-roa-onlinematching} when substituting the role of water-level with normalized water-level. Then following lemmas hold.
\begin{lemma}[Dominance Property]\label{lemma1}
For fixed values of $\mathbf{Z}_{-j}$, 
\begin{equation}
\mathbb{E}\{\beta_j| \mathbf{Z}_{-j}\}\geq  b_{i,j}\Big(\int_{0}^{1}g(z)dz-\sum_{k=1}^{r}\theta_k \int_{Z_{a_k}}^{Z_{a_{k+1}}}g(z)\, dz\Big)
\end{equation}
\end{lemma}
\begin{proof}
If $Z_j\in[Z_{a_k},Z_{a_{k+1}})$ for some $k$, then normalized water level of $i$ at the time that $j$ arrives is $\theta_k$. Let $\{x_{i',j}\}_{i'\in N(j)}$ denotes the allocation at this step. After the algorithm finishes processing $j$, either the budget of $i$ is exhausted, or the item has been fully allocated without using the whole budget of $i$. In the first case, $\sum_{i'\in N(j)}{b_{i',j}x_{i',j}}\geq B_i (1-\theta_k)$ as the algorithm is greedy. In the second case, $\sum_{i'\in N(j)}{b_{i',j}x_{i',j}}\geq{b_{i,j}}$, as greedy algorithms continuously allocate the item to the bidder with maximum bid among those with non-zero remaining budget, which has a bid at least equal to $b_{i,j}$ as $i$'s budget is not yet exhausted during processing $j$. Hence,
\begin{equation}
\beta_j\geq \min \{ B_i (1-\theta_k),{b_{i,j}}\} g(Z_j) \geq b_{i,j}(1-\theta_k) g(Z_j)
\end{equation}
where the last inequality is true because $B_i\geq  b_{i,j}$. So,
\begin{align}
&\label{4-1}\mathbb{E}\{\beta_j \mathds{1}(Z_j\in [0,Z_{a_{1}})|\mathbf{Z}_{-j}\}=b_{i,j}\int_{0}^{Z_{a_1}}g(z)\, dz,\\
&\label{4-2}\forall k\leq r:\mathbb{E}\{\beta_j \mathds{1}(Z_j\in [Z_{a_{k}},Z_{a_{k+1}})|\mathbf{Z}_{-j}\}\geq b_{i,j} (1-\theta_k) \int_{Z_{a_k}}^{Z_{a_{k+1}}}g(z)\, dz~.
\end{align}
Since $\mathbb{E}\{\beta_j| \mathbf{Z}_{-j}\}\geq\sum_{k=0}^{r} {\mathbb{E}\{\beta_j \mathds{1}(Z_j\in [Z_{a_{k}},Z_{a_{k+1}}))|\mathbf{Z}_{-j}\}}$, the lemma follows by summing up the right-hand sides of (\ref{4-1}) and (\ref{4-2}) for $1\leq k\leq r.$\qed
\end{proof}
\begin{lemma} [Dominance Property] \label{lemma2}For fixed values of $\mathbf{Z}$, $\alpha_i\geq \int_{0}^{\theta} 1-g(Z^c(\bar{y})) \, d\bar{y}$~.
\end{lemma}
\begin{proof}
The proof is exactly as the proof of Lemma~\ref{lemma2-onlinematching}, by substituting the water-level with normalized water-level.\qed
\end{proof}
By the same calculations as in the proof of Theorem~\ref{th-roa-onlinematching}, for any $g$ (non-decreasing, satisfying $g(1)=1$) and $F$ that satisfies the integral equation (\ref{inteq1}), the property (2) in Section~\ref{sec-framework} holds by an application of the above lemmas, i.e. $\mathbb{E}\{b_{i,j}\alpha_i+\beta_j\}\geq F b_{i,j}$.  The proof is then completed by setting $g(x)=e^{x-1}$ and $F=\int_0^{1}g(x)dx=\frac{e-1}{e}.$\qed
\end{proof}
\subsubsection{Integral greedy algorithms in the infinitesimal-bid limit}
In this section, we obtain the result of \citep{GM08} for the integral greedy algorithm using our dual-fitting approach. The integral greedy algorithm is very simple to describe. Upon arrival of a new item, the algorithm selects the maximum bidder among neighbors whose budgets are not yet exhausted, and assigns the item to this bidder if possible (i.e., if the remaining budget is at least the bid value), and otherwise does nothing. Tie-breaking is done using a fixed ordering on the bidders. Let us denote the greedy algorithm in \citep{GM08} by \texttt{I-greedy}, and let \texttt{AM-greedy} be an arbitrary allocation-monotone greedy (fractional) algorithm whose decisions match those of \texttt{I-greedy} whenever \texttt{I-greedy} is able to assign the item. We have this lemma, with proof provided in the appendix.
\begin{lemma}
\label{lemma-igreedy}
There exists a randomized dual solution for OnBAP with random arrival s.t.
\begin{itemize}
\item The objective value of \texttt{I-greedy} is at least the dual objective divided by
$1+\max\{\frac{b_{i,j}}{B_i}\}$,
\item The dual becomes feasible in expectation when 
divided by $\frac{e-1}{e}$.
\end{itemize}
\end{lemma}
This lemma shows that in the infinitesimal-bid limit, \texttt{I-greedy} becomes $(\frac{e-1}{e})$-competitive under random arrival order, as desired.
\section{Results for online general assignment problem}
\label{sec-ongap}
So far we have discussed OnBAP. There is a more general problem, the Online Generalized Assignment Problem (OnGAP), which includes all online allocation problems discussed in this paper as special cases. OnGAP is similar to OnBAP except we have different weights and bids when selling an item to a buyer. The standard primal-dual linear programming formulation of OnGAP and its dual are as follows.
\begin{align}
&\text{maximize~~} \quad \sum_{(i,j)\in E}{b_{i,j}x_{i,j}}~~~\text{s.t.}\quad
&\text{minimize~~} \sum_{i\in L }{B_i\alpha_i}+\sum_{j\in R}{\beta_j}~~~\text{s.t.} \nonumber\\
&\sum_{j\in N(i)}{w_{i,j}x_{i,j}}\leq B_i , \quad i\in L\quad
&w_{i,j}\alpha_i+\beta_j\geq b_{i,j}\quad (i,j) \in E \nonumber \\
&\sum_{i\in N(j)}{x_{i,j}}\leq 1, \quad j\in R\quad 
&\alpha_i\geq 0,\quad i\in L \nonumber\\
&x_{i,j} \geq 0, \quad (i,j)\in E \quad 
&\beta_j\geq 0,  \quad j\in R \label{lp-ongap}
\end{align}
While it is tempting to design a primal-dual algorithm that achieves a constant worst-case competitive ratio in the fractional OnGAP, we show that no algorithm can be constant-competitive via 
the following theorem, whose proof is left to the appendix. 
\begin{theorem}
\label{thm-ongap1}
For a fixed $\eta \geq 1$, consider instances of OnGAP when $\forall (i,j) \in E: w_{i,j}\leq b_{i,j} \leq \eta w_{i,j}$. Then worst-case competitive-ratio of online algorithms on these instances is bounded above by $O(\frac{1}{{\log (\eta)}})$.
\end{theorem}

Now, we present a method of enhancing our previous constant-competitive algorithms for OnBAP to be applicable in OnGAP, under both worst-case and random order arrival models. We have the following theorem, which shows the upper-bound in Theorem~\ref{thm-ongap1} is tight under worst-case model. We leave the proof to the Appendix.
\begin{theorem}
\label{thm-ongap2}
Suppose $\texttt{ALG}$ is a $c$-competitive online algorithm for OnBAP under either the worst-case or random order arrival model. Now, consider the following randomized online algorithm under the mode in which $\texttt{ALG}$ is $c$-competitive:
\begin{itemize}
\item At initialization time, sample $s\in \{0,1,\ldots, \lfloor \log(\eta) \rfloor\}$ uniformly at random. (The base of the logarithm is $2$.)
\item Once $j$ arrives, discard the bids $b_{i,j}$ s.t. $\lfloor \log(b_{i,j}/w_{i,j}) \rfloor\neq s$
\item Run $\texttt{ALG}$ on undeleted $(i,j)$ pairs assuming $\text{bids}=\text{weights}=\{w_{i,j}\}$.
\end{itemize}
This algorithm achieves a competitive ratio of $\frac{c}{2(1+ \lfloor \log(\eta) \rfloor)}$ in expectation.
\end{theorem}

While the proposed algorithm in Theorem~\ref{thm-ongap2} is a randomized algorithm, the randomization happens only in the initialization step. Thus, it can easily be derandomized by running each of the $1+\lfloor \log(\eta) \rfloor$ possible versions of the randomized algorithm (i.e., one for each choice of the initialization step) in parallel, and outputting a fractional allocation for each item $j$ which is the unweighted average of the fractional allocations computed by the different versions of the randomized algorithm.

\bibliography{online-alloc}
\bibliographystyle{plain}
\section*{Appendix}
\subsection*{Proof of Theorem~\ref{th4}}
The proof is based on primal-dual method. Suppose we construct a dual solution as follows. Initialize $\alpha_i=0,i\in L$ and $\beta_j=0,j\in R$. Once $j$ arrives, update the dual variables as follows, where $l_j$ denotes the final water-level of vertices whose water-level has been increased during processing $j$:
\begin{align}
&\alpha_i=G(y_{i,\text{new}})/F,~~~ \forall i \in N(j),\nonumber\\
&\beta_j=1-G(l_j)/F=1-G(y_{i,\text{new}})/F,\forall i \in L(j).
\end{align}
in which $G(t)\triangleq \int_0^t g(x)dx$ for $t\in [0,1]$. Following the fact that water-filling raises the minimum water-level continuously, $\alpha_i+\beta_j\geq G(l_j)/F+1-G(l_j)/F=1$ at the time that the algorithm finishes processing $j$, for every $i\in N(j)$. After that $\alpha_i$ doesn't decrease and hence the dual solution is always feasible. To compare the value of dual and primal, first suppose $G$ has continuous monotone non-decreasing second derivative in $[0,1]$ (i.e. $g'$ is continuous and non-decreasing). Hence we have the following Taylor approximation,
\begin{equation}
\forall x_1\leq x_2\in[0,1],x_1\leq x_2:G(x_2)=G(x_1)+g(x_2)(x_2-x_1)-\tfrac12g'(\delta)(x_2-x_1)^2
\end{equation}
for some $\delta\in[x_1,x_2]$.  Now, look at the time that $j$ arrives. There are two cases. In the first case, the arriving item $j$ will be fully matched ($(\Delta\text{primal})_j$=1), and so we have
\begin{align}
(\Delta\text{dual})_j&=1-G(l_j)/F+1/F\sum_{i\in L(j)}(G(y_{i,\text{new}})-G(y_{i,\text{old}}))\nonumber\\
&\overset{(1)}{=}1/F-g(l_j)/F+g(l_j)/F\sum_{i\in L(j)}{(y_{i,\text{new}}-y_{i,\text{old}})}-\tfrac12g'(\delta)/F\sum_{i\in L(j)}{(y_{i,\text{new}}-y_{i,\text{old}})^2}\nonumber\\
&\overset{(2)}{\leq}1/F-\tfrac12g'(0)/F\sum_{i\in L(j)}{x_{i,j}^2}\overset{(3)}{\leq} 1/F-\frac{g'(0)}{2F\lvert L(j)\rvert}(\sum_{i\in L(j)}{x_{i,j}})^2\nonumber\\
&\overset{(4)}{\leq} \frac{(\Delta\text{primal})_j}{F}-\frac{g'(0)}{2dF}(\Delta\text{primal})_j^2
\end{align}
in which equality (1) holds due to integral equation (\ref{inteq1}), inequality (2) comes from the monotonicity of $g'$, inequality (3) is coming from the well-known Cauchy-Schwarz inequality, and inequality (4) is true because the degree of each online vertex is at most $d$.
In second case, $(\Delta\text{primal})_j<1$ (and hence $l_j=1$). So we have
\begin{align}
(\Delta\text{dual})_j&=1-G(1)/F+1/F\sum_{i\in L(j)}(G(1)-G(y_{i,\text{old}}))\nonumber\\
&\overset{(1)}{=}g(1)/F\sum_{i\in L(j)}{(1-y_{i,\text{old}})}-\tfrac12g'(\delta)/F\sum_{i\in L(j)}{(1-y_{i,\text{old}})^2}\nonumber\\
&\leq \frac{(\Delta\text{primal})_j}{F}-\frac{g'(0)}{2dF}(\Delta\text{primal})_j^2
\end{align}
in which equality (1) holds due to the integral equation (\ref{inteq1}) and the rest as in the first case. From both cases, we can conclude the following.
\begin{equation}
\text{dual}\leq \frac{\text{primal}}{F}-\frac{g'(0)}{2dF}\sum_{j\in R}(\Delta\text{primal})_j^2
\label{band1}
\end{equation}	
  Now, let's define a partitioning of $R$ as following.  $J_1\triangleq\{j\in R: (\Delta\text{primal})_j \geq 1/2\}$ and $J_2\triangleq\{j\in R: (\Delta\text{primal})_j < 1/2\}$. By the fact that water-filling is greedy, no $i\in L$ can receive water from a vertex in $J_2$ until its water-level is at least $1/2$. So, $\text{primal}=\sum_{j\in R}(\Delta\text{primal})_j\leq 2 \sum_{j\in J_1}(\Delta\text{primal})_j$. Hence, we have
\begin{equation}
\label{band2}
\sum_{j\in R}(\Delta\text{primal})_j^2\geq \sum_{j\in J_1}(\Delta\text{primal})_j^2\geq \tfrac12\sum_{j\in J_1}(\Delta\text{primal})_j\geq \tfrac14 (\text {primal})
\end{equation}
where the second inequality is true by definition of $J_1$. The theorem is then proved by combining (\ref{band1}) and (\ref{band2}), and substituting $g(x)=e^{x-1}$ and $F=\frac{e-1}{e}$.\qed
\subsection*{Proof of Lemma~\ref{lemma-igreedy}}
For an input sequence $\textbf{seq}$, define $\hat{\textbf{seq}}$ to be $\textbf{seq}$ with all elements $j$ deleted if item $j$ was not assigned to the maximum bidder in \texttt{I-greedy}($\textbf{seq}$). Clearly, 
\texttt{I-greedy}$(\textbf{seq})\equiv$ \texttt{I-greedy}$(\hat{\textbf{seq}})\equiv$ \texttt{AM-greedy}$(\hat{\textbf{seq}})$ by definition. Let $\{\hat{\alpha}_i,\hat{\beta_j}\}$ be the set of dual variables constructed as in the proof of Theorem~\ref{th3} for \texttt{AM-greedy}$(\hat{\textbf{seq}})$. Now, define the following dual variables for the input sequence $\textbf{seq}$.
\begin{equation}
\alpha_i=\hat{\alpha}_i, \forall i \in L,~~~~\beta_j= \left\{ \begin{array}{ll}
\hat{\beta}_j &\mbox{ if $j \in \textbf{seq}$} \\
(\frac{e-1}{e})b_{i,j} &\mbox{ otherwise}
\end{array} \right.
\end{equation}
As $\{\hat{\alpha}_i,\hat{\beta_j}\}$ is $\frac{e-1}{e}$-feasible in expectation for the input sequence $\hat{\textbf{seq}}$, $\{{\alpha}_i,{\beta_j}\}$ is $\frac{e-1}{e}$-feasible in expectation for the input sequence ${\textbf{seq}}$. Moreover, the objective value of $\{ {\alpha}_i,{\beta_j}\}$ differs from the objective value of  $\{\hat{\alpha}_i,\hat{\beta_j}\}$ by at most the sum of the maximum bids of those vertices $j$ that haven't been allocated in \texttt{I-greedy}$({\textbf{seq}})$, which happens only if the budget of the maximum neighboring bidder of $j$ is almost exhausted, i.e. $y_{i}+b_{i,j}>B_i$ in which $i$ is the maximum bidder. Hence, one can conclude that the objective value of $\{{\alpha}_i,{\beta_j}\}$ is at most equal to $(1+\max {\frac{b_{i,j}}{B_i}})$ times the objective value of $\{\hat{\alpha}_i,\hat{\beta_j}\}$. The lemma then follows as the objective value of $\{\hat{\alpha}_i,\hat{\beta_j}\}$ is equal to the objective value of \texttt{I-greedy}$(\textbf{seq})$.\qed

\subsection{Proof of Theorem~\ref{thm-ongap1}}
Pick any arbitrary algorithm, whose objective value is denoted by $\mathtt{ALG}$, and consider the following input. There is one buyer in the offline side, and $n=2^{k}-1$ items in the online side, for some $k\in \mathbb{N}$. The ordered sequence of items are partitioned into $k$ bundles $0,1,\ldots,k-1$.  Bundle $t$ has $2^t$ items, all with the same bid of $1$ and weight of $2^{-s}$. Bundles are arriving in the order of their indices, from $0$ to $k-1$. When algorithm runs on this input, let 
$$c_t=\sum_{j\in \text{bundle}~ t}{w_{1,j}x_{1,j}}.$$
Note that due to feasibility of the algorithm, $\sum_{t=0}^{k-1}{c_t}\leq 1$. Now, run the algorithm on a modified input w.r.t parameter $0\leq  s\leq k-1$. Let bundles $0,1,\ldots, s$ be the same, but bundles $s+1,\ldots,k-1$ have zero bids. Obviously, for this input instance $$\text{OPT}=2^{s},~~\mathtt{ALG}=\sum_{t=0}^{s} c_t 2^t.$$
 If we denote the competitive ratio by $\alpha$, we have $\alpha=\sum_{t=0}^{s} c_t 2^{t-s}$. Now,  we can construct $k$ input sequences for different values of $s$, and calculate the competitive ratio for each of them. Hence, we have the following factor revealing linear program that upper-bounds the competitive ratio of any algorithm.
 \begin{align}
 \label{fact-lp}
& \max {\alpha}~~\text{s.t.}\nonumber\\
&\sum_{t=0}^{k-1}{c_t}\leq 1,\nonumber \\
 &\alpha \leq \sum_{t=0}^{s} c_t 2^{t-s},~~ s\in\{0,1,\ldots,k-1\},\nonumber \\
&\alpha\geq 0,\{ c_t\}\geq 0
 \end{align}
Summing up the left hand side and right hand side of the second set of constraints in  (\ref{fact-lp}) for different values of $s$, we have
$$k \alpha\leq \sum_{s=0}^{k-1}\sum_{t=0}^{s} c_t 2^{t-s}\leq 2\sum_{s=0}^{k-1}c_s\leq 2$$
in which the second inequality is coming from re-arranging the terms in the sum, and the last inequality holds as $\sum_{t=0}^{k-1}{c_t}\leq 1$. Hence, $\alpha^*\leq \frac{2}{k}$. This completes the proof, since in all of the input sequences, $1\leq\frac{b_{i,j}}{w_{i,j}}\leq 2^k$.\qed
\subsection{Proof of Theorem~\ref{thm-ongap2}}
Obviously, the output of the algorithm is feasible. To show the competitive ratio, partition the edge set $E$ into $E_0,E_1,\ldots,E_{\lfloor\log(\eta)\rfloor}$ s.t.
$$ (i,j)\in E_s \iff \lfloor\log(\tfrac{b_{i,j}}{w_{i,j}})\rfloor=s$$
Furthermore, suppose $\{x_{i,j}\}$ is the algorithm's allocation and $\{x^*_{i,j}\}$ is the optimal offline allocation of the original OnGAP. Hence, $\text{OPT}=\sum_{(i,j)\in E}{b_{i,j}x^*_{i,j}}$.
Now, conditioned on $s$ look at the remaining bids in $E_s$. We have 
\begin{equation}
(i,j)\in E_s \iff s\leq \log(\tfrac{b_{i,j}}{w_{i,j}}) < s+1 \iff 2^{s}\leq \tfrac{b_{i,j}}{w_{i,j}} < 2^{s+1}
\end{equation}
Define $\tilde{b}_{i,j} \triangleq 2^s w_{i,j}$. We have the following lemma.
\begin{lemma} Conditioned on sampling $s$ in the first step, our algorithm gets
$$\mathbb{E}\{\sum_{(i,j)\in E}{b_{i,j}x_{i,j}}| s\}\geq \frac{c}{2}\sum_{(i,j)\in E_s}{b_{i,j}x^*_{i,j}}$$.
\end{lemma}
\begin{proof} We have
$$\mathbb{E}\{\sum_{(i,j)\in E}{b_{i,j}x_{i,j}}| s\} \geq \mathbb{E}\{2^s\sum_{(i,j)\in E}{w_{i,j}x_{i,j}}| s\}\geq \mathbb{E}\{2^s c\sum_{(i,j)\in E_s}{w_{i,j}x^*_{i,j}}| s\}=c \sum_{(i,j)\in E_s}{\mathbb{E}\{\tilde{b}_{i,j}x^*_{i,j}}| s\}$$
and the proof then follows as $\tilde{b}_{i,j}\geq \frac{1}{2}b_{i,j}$. \qed
\end{proof}
To complete the proof of Theorem~\ref{thm-ongap2} note that $\text{OPT}=\sum_{s}{\sum_{(i,j)\in E_s}}{b_{i,j}x^*_{i,j}}$. Hence we get
\begin{align}
\mathbb{E}\{\sum_{(i,j)\in E}{b_{i,j}x_{i,j}}\}&=\sum_{s}\Pr\{s\}\mathbb{E}\{\sum_{(i,j)\in E}{b_{i,j}x_{i,j}}| s\}\geq \frac{1}{1+\lfloor \log(\eta)\rfloor}\sum_{s}{\frac{c}{2}\sum_{(i,j)\in E_s}{b_{i,j}x^*_{i,j}}}\nonumber\\
&=\frac{c}{2(1+\lfloor \log(\eta)\rfloor)}\text {OPT}
\end{align}
which completes the proof.\qed
\end{document}